\documentclass{cccg18}
\usepackage{graphicx,amssymb,amsmath}

\usepackage{gensymb}
\usepackage[ruled,vlined]{algorithm2e}

\title{A Linear-Time Approximation Algorithm for the Orthogonal Terrain Guarding Problem}

\author{Wei-Yu Lai\thanks{Department of Computer Science and Information Engineering,  National Taiwan University of Science Technology, {\tt D10115005@mail.ntust.edu.tw}}
        \and
        Tien-Ruey Hsiang\thanks{Department of Computer Science and Information Engineering,  National Taiwan University of Science Technology, {\tt  trhsiang@csie.ntust.edu.tw}}}

\index{Lai, Wei-Yu}
\index{Hsiang, Tien-Ruey}

\begin{document}
\thispagestyle{empty}
\maketitle

\begin{abstract}
In this paper, we consider the 1.5-dimensional orthogonal terrain guarding problem. In this problem, we assign an $x$-monotone chain $T$ because each edge is either horizontal or vertical, and determine the minimal number of vertex guards for all vertices of $T$. A vertex $v_i$ sees a point $p$ on $T$ if the line segment connecting $v_i$ to $p$ is on or above $T$. 
We provide an optimal algorithm with $O(n)$ for a subproblem of the orthogonal terrain guarding problem. In this subproblem, we determine the minimal number of vertex guards for all right(left) convex vertices of $T$.  
Finally, we provide a 2-approximation algorithm that solves the  1.5-dimensional orthogonal terrain guarding problem in $O(n)$ time.  
\end{abstract}

\section{Introduction}
A 1.5-dimensional(1.5D) terrain $T$ is an $x$-monotone polygonal chain specified by $n$ vertices $V(T)=\{v_1,...,v_i=(x(v_i),y(v_i)),...,v_n\}$ ordered from left to right, such that $x(v_i) \leq x(v_{i+1})$ (strict monotonicity is often assumed). The vertices induce $n$-1 edges $E(T)=\{e_1,...,e_i =(v_i,v_{i+1}),...,e_{n-1}\}$.
Terrain $T$ is called an orthogonal terrain if each edge $e \in E(T)$ is horizontal or vertical.

A point $p$ sees $q$ (and $q$ sees $p$) if the line segment $\overline{pq}$ lies above $T$, or more precisely, if it does not intersect the open region that is bounded from above by $T$ and from the left and right by the downwards vertical rays emanating from $v_1$ and $v_n$.

There are two types of terrain guarding problems. One is the continuous terrain guarding problem, the objective of which is to determine a minimum cardinality subset of $T$ that guards $T$. The other is the discrete terrain guarding problem, where the goal is to guard candidate set $G$ and witness set $W$ on $T$, the objective is to determine a minimum cardinality subset of $G$ that guards $W$. So far, the study of the orthogonal terrain guarding problem has focused on solving discrete terrain guarding problems, where guarding candidate set $G$ and witness set $W$ are the vertices of the terrain.   

\subsection{Related works}

Let a terrain $T$ denote an $x$-monotone chain. Some researchers have considered an orthogonal terrains. A vertex $v_i$ of a orthogonal $T$ is convex (reflex) if the angle formed by the edges $e_{i-1}$ and $e_i$ above $T$ is of 90\degree (270\degree). An orthogonal terrain distinguishes between two types of convex vertices: a left convex and a right convex. A convex vertex is a left (right) convex vertex if  $e_{i-1}$($e_i$) is vertical. 

Katz and Roisman~\cite{2-OTG} provided a 2-approximation algorithm for the problem of guarding the vertices of an orthogonal terrain. The researchers constructed a chordal graph that reveals the relationship of visibility between vertices. An orthogonal terrain guarding problem can be formulated as a 2-approximation algorithm by using the algorithm to compute the minimum clique cover of a chordal graph\cite{F.Gavril}.

Durocher et al.~\cite{O-OTG} suggested a linear-time algorithm to guard the vertices of an orthogonal terrain under a directed visibility model. Directed visibility mode considers various types of vertices with different visibilities. If $u$ is a reflex vertex, $u$ sees a vertex $v$ of $T$ if and only if every point in the interior of the line segment $\overline{uv}$ lies strictly above $T$. If $u$ is a convex vertex, then $u$ sees a vertex $v$ of $T$ if and only if $\overline{uv}$ is a non-horizontal line segment that lies on or above $T$.

Lyu and {\" U}ng{\"o}r~\cite{nlogm} formulated a 2-approximation algorithm for the orthogonal terrain guarding problem that runs in $O(n$log$m)$, where $m$ is the output size. They provided an optimal algorithm for the subproblem of the orthogonal terrain guarding problem. 
Based on the type of vertex of orthogonal terrain,  
the subproblem with the objective of finding a minimum cardinality subset of $V(T)$ that guards all right(left) convex verteices of $V(T)$. The optimal algorithm  can be performed through stack operation to reduce time complexity. The $O(n$log$m)$ time 2-approximation algorithm was previously considered the best algorithm for the orthogonal terrain guarding problem.  

King and Krohn ~\cite{NP} employed a 1.5D terrain to prove that the general terrain guarding problem is NP-hard through planar 3-SAT.

In early studies of the 1.5D terrain guarding problem, the algorithm design of the constant-factor approximation algorithm was discussed.

Ben-Moshe et al.~\cite{F-app} provided the first constant-factor approximation algorithm for the terrain guarding problem and left the complexity of the problem open.
King~\cite{5-app} devised a simple 4-approximation algorithm that was later determined to be a 5-approximation algorithm.
In 2011, Elbassioni et al.~\cite{4-app} formulated a 4-approximation algorithm.

Gibson et al.~\cite{D-PTAS} considered the discrete terrain guarding problem that determines the minimal cardinality from a candidate point to the target point set. The researchers proved the existence of a planar graph that appropriately relates the local and global optimum, thereby indicating that the discrete terrain guarding problem enables the application of a polynomial time approximation scheme (PTAS) based on a local search.

For the continuous 1.5D terrain guarding problem, Friedrichs et al.~\cite{C-PTAS} constructed finite guard and witness sets $G$ and $X$ such that there existed an optimal guard cover $G’' \subseteq G$ that covers terrain $T$, and that when these guards monitored all points in $X$, the entire terrain was guarded.  The continuous 1.5D terrain guarding problem applied a PTAS through the construction of a finite guard and witness set and former PTAS\cite{D-PTAS}.

\subsection{Our results and problem definition}
In this paper, we describe a 2-approximation algorithm that runs in $O(n)$ time. This algorithm is an improvement on what was previously considered the best algorithm~\cite{nlogm}, namely a 2-approximation algorithm with $O(n$log$m)$ running time where $n$ and $m$ are the size of the input and output, respectively. The orthogonal terrain guarding problem is defined as follows.

{\bf Definition 1} (orthogonal terrain guarding problem). Given an orthogonal terrain $T$, compute a subset $G$$\subseteq$$V(T)$ of minimum cardinality that guards $V(T)$.

\subsection{Paper organization}
The remainder of this paper is organized as follows: Section 2 describes the preliminaries; Section 3 provides an optimal algorithm for the right (left) convex vertex guarding problem, as well as its proof;
Section 4 reveals the $O(n)$ time 2-approximation algorithm for the orthogonal terrain guarding problem; and Section 5 presents our conclusions.

\section{Preliminaries}

Given an orthogonal terrain $T$, we denote the $x$- and $y$-coordinates of a vertex $v \in V(T)$ by $x(v)$ and $y(v)$, respectively, as well as the rightmost(leftmost) vertex of $v$ that sees $v$ by $R(v)$($L(v)$). The symbol $vis(v)$ denotes the visibility region of $v$ with $vis(v)=\{u \in v(T) \mid v$ sees $u$\}.

$V(T)$ can be broken down into the right convex vertex, left convex vertex, right reflex vertex and left reflex vertex.  
A vertex $v_i$ of an orthogonal terrain $T$ is convex (reflex) if the angle formed by the edges $e_{i-1}$ and $e_i$ above $T$ is 90\degree(270\degree).
A convex vertex is left (right) convex if $e_{i-1} (e_i)$ is vertical.
Also, if the leftmost (or rightmost) vertex of $T$ is an endpoint of a horizontal segment, it is marked as a convex vertex. 
We denote the set of left convex vertices by $V_{LC}(T)$,
the set of right convex vertices by $V_{RC}(T)$,
the set of left reflex vertices by $V_{LR}(T)$
and the set of right reflex vertices by $V_{RR}(T)$.
Figure~\ref{fig1} illustrates an example of convex and reflex vertices of a terrain. 

\begin{figure}[h]
\begin{center}
\includegraphics[scale=0.5]{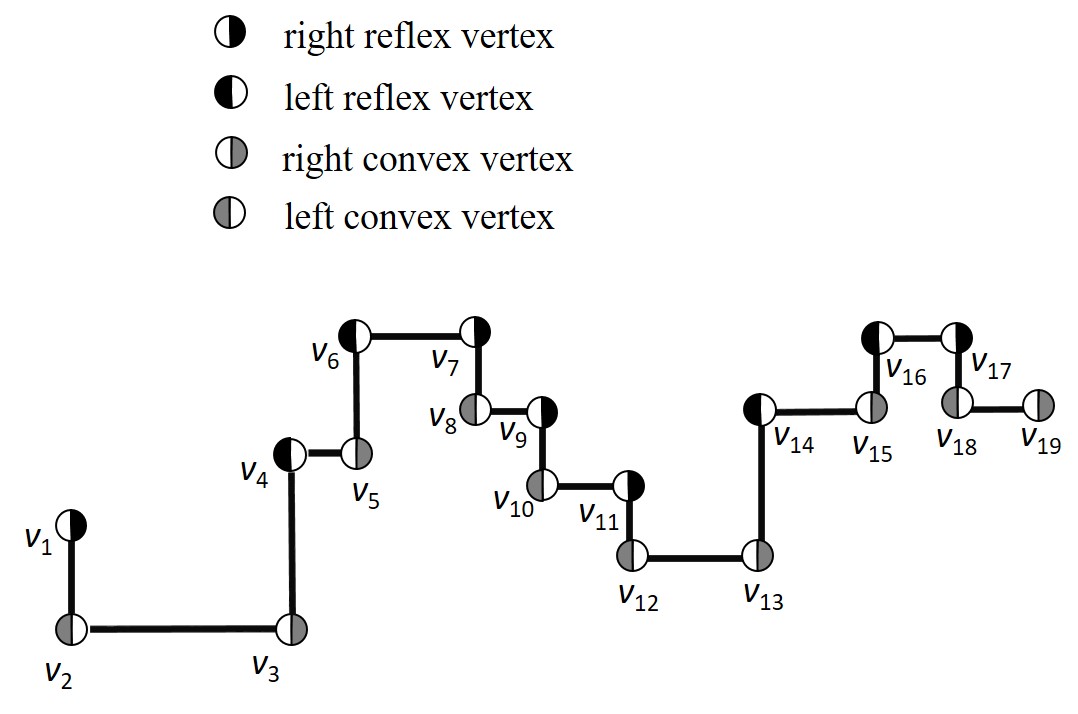}
\caption{Vertex types of the orthogonal terrain.}
\label{fig1}
\end{center}
\end{figure}

Using Figure~\ref{fig1}, an easy observation are made on orthogonal terrains:


\begin{obs}
\label{lrv sees two rcv}
A vertex $v \in V_{LR}(T)$ sees at most two right convex vertices $v_j$ and $v_k$, one of which is immediately to  the right of $v$, and the other is immediately below $v$. 
\end{obs}

For example, $v_{14}$ is a left reflex vertex that sees $v_{13}$ and $v_{15}$, but it cannot see $v_3, v_{5}$ or $v_{18}$. 

\begin{obs}
\label{lmost RC see a LR}
The leftmost $v$ in $V_{RC}(T)$ sees at most one left reflex vertex, which is immediately above $v$.
\end{obs}

In Figure~\ref{fig1}, $v_3$ is the leftmost vertex in $V_{RC}(T)$ and $v_3$ sees $v_4 \in V_{LR}(T)$. If $v_3$ is seen by any $v_p \in V_{LR}(T)$ such that $x(v_p) \neq x(v_3)$, $v_2$ must be the left reflex vertex and $v_1$ must be the right convex vertex; however, this contradicts  $v_3$ being the leftmost right convex vertex and $v_1$ being the right convex vertex.
  



\begin{obs}
\label{vj low Lvi}
Let $v$ be a right convex vertex and $u$ be any vertex such that $x(L(v))<x(u)<x(v)$, we have $y(u) \leq y(L(v))$.
\end{obs}

Take Figure~\ref{fig1} for example, $v_{13}$ is a right convex vertex, $v_7$ is the leftmost vertex in $vis(v_{13})$ and every vertex of $\{v_8, v_9,...,v_{12} \}$ is not above $v_7$.  
Consider, for Observation~\ref{vj low Lvi},
if $v$ be a left reflex vertex and $x(L(v))<x(u)<x(v)$, $u$ can be located above $L(v)$. 
For example, $v_6$ is a left reflex vertex and $v_1$ sees $v_6$, but $v_4$ is above $v_1$.


In addition, regarding the 1.5D terrain guarding problem, Ben-Moshe~\cite{F-app} derived a visibility relation of ordered vertices that was also employed in a few studies ~\cite{5-app,4-app,D-PTAS,2-OTG,O-OTG}. Their lemma is stated as follows: 

\begin{lemma}[\cite{F-app}]
\label{order lemma}
Let $a, b, c, d \in V(T)$ be four vertices of a terrain $T$ such that $x(a)<x(b)<x(c)<x(d)$. If $a$ sees $c$ and $b$ sees $d$, $a$ sees $d$.
\end{lemma}

Other studies on the orthogonal terrain guarding problem have propounded some visibility relations, such as  



\begin{lemma}[\cite{2-OTG}]
\label{left side}
If $v \in V_{RR}(T)$ sees $u \in V_{RC}(T)$, $v$ on the left side of $u$. 
\end{lemma}

\begin{lemma}[\cite{O-OTG}]
\label{higher can't see}
If $v \in V_{RC}(T) \cup V_{LC}(T)$ is higher than $u \in V_{RR}(T) \cup V_{LR}(T)$, $u$ cannot see $v$.
\end{lemma}

\section{Optimal algorithm for the right (left) convex vertex guarding problem}
In this section, we propose an optimal algorithm for the right (left) guarding problem. 
The optimal algorithm for the right(left) guarding problem can be expressed in the from of a 2-approximation algorithm for the orthogonal terrain guarding problem. 
We design the algorithm through a visibility relationship between reflex vertices, and then prove that the output is the optimal solution. 
We first define the right(left) convex vertex guarding problem. 

{\bf Definition 2} (right convex vertex guarding problem). Given an orthogonal terrain $T$,  a subset $G_r$$\subseteq$$V_{LR}(T) \cup V_{RR}(T)$ of minimum cardinality that guards $V_{RC}(T)$ is computed.

{\bf Definition 3} (left convex vertex guarding problem). Given an orthogonal terrain $T$, a subset $G_l$$\subseteq$$V_{LR}(T) \cup V_{RR}(T)$ of minimum cardinality that guards $V_{LC}(T)$ is computed.

\begin{lemma}
\label{L(vi) not see}
If $ v_i\in V_{RC}(T)$ and $y(L(v_i)) < y(v_{i+1})$, then $L(v_i)$ cannot guard $v_j \in V_{RC}(T)$ that is on the right side of $v_{i+1}$.

\end{lemma}
 
\begin{proof}
For $j>i+1$, we consider two cases $v_j \in V_{RC}(T)$ into $y(v_j) > y(v_{i+1})$ and $y(v_j) \leq y(v_{i+1})$.
\newtheorem{case}{Case}
\begin{case}
 $y(v_j) > y(v_{i+1})$
\end{case}
If $y(v_j) > y(v_{i+1})$, we know that $L(v_i)$ cannot guard $v_j$ when $j>i+1$ based on Lemma~\ref{higher can't see}. 
\begin{case}
$y(v_j) < y(v_{i+1})$
\end{case}
If $y(v_j) < y(v_j)$, line $\overline{L(v_i)v_j}$ and line $\overline{v_iv_{i+1}}$ have an intersection point below $v_{i+1}$; therefore, $L(v_i)$ cannot guard $v_j$.
\end{proof}

\begin{lemma}
\label{OPT lemma}
An optimal solution for the right convex guarding problem includes the highest vertex $v_j \in vis(v_i)$ such that $v_i$ is the leftmost right convex vertex.  
\end{lemma}

\begin{proof}
We discuss the visibility relationship between reflex vertices $v_k \in S=
vis(v_i)\cap \{V_{LR}(T) \cup V_{RR}(T) \}$
to prove Lemma~\ref{OPT lemma}.
The set $S$ contains a left reflex vertex and right reflex vertices by Lemma~\ref{left side} and Observation~\ref{lmost RC see a LR}.
Finally, we can prove the highest vertex $v_j \in S$ such that $vis(v_j)\cap V_{RC}(T) = vis(S) \cap V_{RC}(T)$. 

First, we discuss the visibility relationship between vertices in $S_L = vis(v_i) \cap V_{RR}(T)$.
We know that $vis(L(v_i))\cap V_{RC}(T) = vis(S_L)\cap V_{RC}(T)$ based on Lemma~\ref{order lemma},
and $L(v_i)$ is the highest vertex in the set $S_L$ based on Observation~\ref{vj low Lvi}, 
as illustrated in Figure~\ref{LR-RC}.  

\begin{figure}[h]
\begin{center}
\includegraphics[scale=0.5]{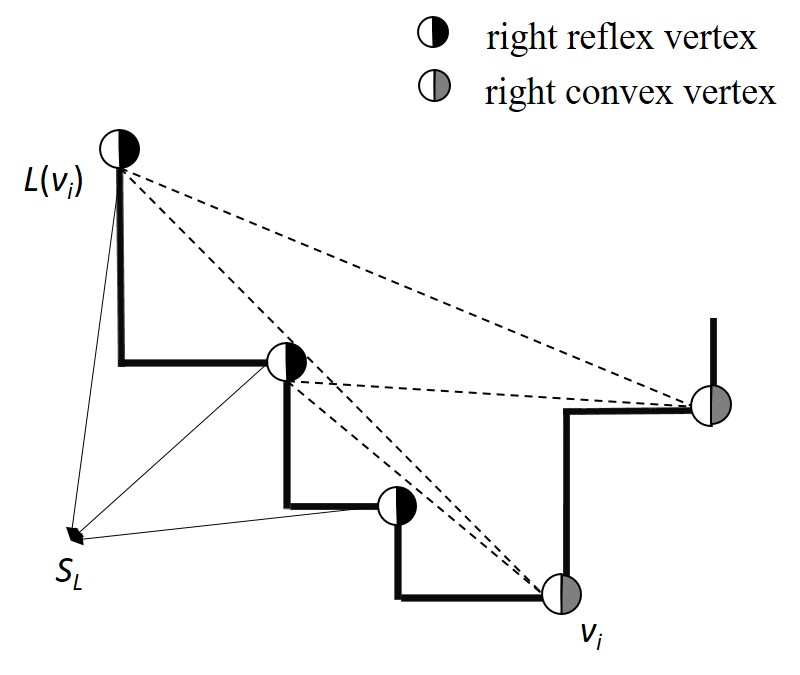}
\caption{The vertex $L(v_i)$ is highest vertex in the set $S_L$ and $vis(L(v_i))\cap V_{RC}(T) = vis(S_L)\cap V_{RC}(T)$.}
\label{LR-RC}
\end{center}
\end{figure}

Next, according to Observation~\ref{lmost RC see a LR}, we discuss the visibility relationship between $L(v_i)$ and $v_{i+1}$.
According to Lemma~\ref{order lemma} and Observation~\ref{lrv sees two rcv}, 
if $y(L(v_i)) \geq y(v_{i+1})$, $vis(L(v_i)) \cap V_{RC}(T) \supseteq vis(v_{i+1}) \cap V_{RC}(T)$.
On the contrary, if $y(v_{i+1}) > y(L(v_i))$, $vis(v_{i+1}) \cap V_{RC}(T) \supseteq vis(L(v_i)) \cap V_{RC}(T)$ based on Lemma~\ref{L(vi) not see}. 

Finally, we know that if $v_j$ is the highest vertex in the set $S$, $vis(v_j) \cap V_{RC}(T) = vis(S) \cap V_{RC}(T)$.
\end{proof}

Based on Lemma~\ref{OPT lemma}, we propose Algorithm~\ref{A1} to compute a $G_r \subseteq V_{LR}(T) \cup V_{RR}(T)$ of minimum cardinality that guards $V_{RC}(T)$. In this paper, we  confine our attention to the number of  $|G_r|$.

\begin{algorithm}[h]
\label{A1}
  \SetAlgoNoEnd
  \caption{Compute a $G_r \subseteq V_{LR}(T) \cup V_{RR}(T)$ of minimum cardinality that guards the $V_{RC}(T)$.}
  \KwIn{$T$: terrain}
  \KwOut{$G_r$}
  
  Compute all $L(v_i)$ for $v_i \in V_{RC}(T)$;
  
  $G_r$ is null;
  
  \For{$v_i \in T$ processed from left to right }
  {
    \If{$v_i \in V_{RC}(T)$ is not guarded by $G_r$}{
      $G_r \leftarrow$ a highest vertex in $\{ L(v_i), v_{i+1} \}$
    }
  }
  \Return $G_r$
\end{algorithm}

\begin{theorem}
\label{gr min}
 The solution $G_r$ of Algorithm~\ref{A1} is an optimal solution for right convex vertex guarding problem. 
\end{theorem}

\begin{proof}
In each iteration of Algorithm~\ref{A1}, we added the leftmost unguarded right convex vertex $v_i$ to $X=\{ x_1, x_2,..., x_b \} \subseteq V_{RC}(T)$. According to line 4 of Algorithm~\ref{A1}, we know that $G_r= \{ g_1,g_2,...,g_b \}$, $x_i \in X$ is guarded by $g_i\in G_r$ and $|G_r| = |X|$. Assuming that $x_i, x_j \in X$ and $i < j$, we know that $vis(x_i) \cap vis(x_j) \cap \{ V_{RF}(T) \cup V_{LF}(T) \} = \emptyset$.
If $vis(x_i) \cap vis(x_j) \cap \{ V_{RR}(T) \cup V_{LR}(T) \} \neq \emptyset$, $x_j \in vis(g_i)$.
Therefore, $G_r$ is the minimum cardinality set that guards $X$ and $vis(G_r)\cap V_{RC}(T) = V_{RC}(T)$.
\end{proof}

\section{An $O(n)$ time 2-approximation algorithm for orthogonal terrain guarding}

In this section, we provide a $O(n)$ time 2-approximation algorithm for orthogonal terrain guarding. This is an improvement over the 2-approximation algorithm with $O(n$log$m)$ running time proposed by Lyu et al.\cite{nlogm}. 
We discuss the approximation ratio and time complexity of Algorithm~\ref{A1} in this section.

First, we discuss the relationship between the orthogonal terrain guarding problem and the right(left) convex vertex guarding problem.
According to \cite{2-OTG} and \cite{nlogm}, the orthogonal terrain guarding problem can be divided into left convex vertex and right convex vertex guarding problems. If we unite the optimal solutions of the left and right convex vertex guarding problems, the union provides a 2-approximation solution for the orthogonal terrain guarding problem.

Next, we demonstrate that Algorithm~\ref{A1} runs in $O(n)$. In Algorithm~\ref{A1}, the point is the time complexity of line 4. Algorithm~\ref{A1} runs in $O(n)$ because line 1 ran in $O(n)$ in~\cite{O-OTG}. Lines 3 and 4 runs $O(n)$ in all algorithms.

\begin{lemma}
\label{angle see}
If $v_i,v_j, v_k \in V(T)$, $v_k$ is the leftmost vertex in the set 
$\{v_{k'} \in vis(v_i) \mid i>{k'}>j \}$
and $\angle v_iv_kv_j \leq 180\degree$, then $v_i$ sees $v_j$.
\end{lemma}
\begin{proof}
Assume that $i>k>j$, $v_k$ is leftmost vertex in the set
$\{v_{k'} \in vis(v_i) \mid i>{k'}>j \}$
and $\angle v_iv_kv_j \leq 180\degree$(the angle $\angle v_iv_kv_j$ is shown in Figure~\ref{fig3}).
If $v_i$ cannot see $v_j$, a vertex $v_p$ exists and lies above the line segment $\overline{v_kv_j}$ and between $v_k$ and $v_j$; however, the assumption that $v_i$ cannot see $v_p$ is contradictory.
\end{proof}

\begin{figure}[h]
\begin{center}
\includegraphics[scale=0.5]{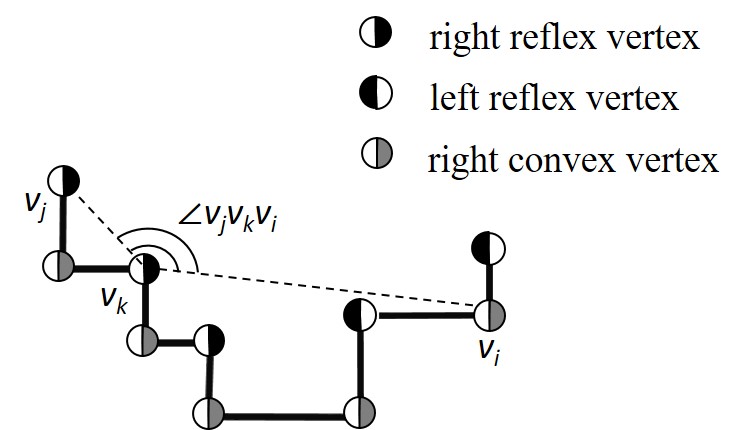}
\caption{Lemma~\ref{angle see}.}
\label{fig3}
\end{center}
\end{figure}





Based on Lemmas~\ref{order lemma},~\ref{angle see} and~\ref{L(vi) not see}, Algorithm~\ref{A1} processes line 4 for all $v_i \in V_{RC}(T)$ in linear time.
If $v_a \in G_r$ guards $v_i$, $g(v_i) = v_a$.
For each $v_i$, we examine whether $v_i$ is guarded by $g_a$ from $v_i$ to $g(v_i)$ such that $x(g(v_i)) \geq x(L(v_i))$.
If $v_j \in V_{RC}(T), j<i$, we do not visit $v_k \in V(T)$ such that $x(g(v_j)) < x(v_k) < x(v_j)$.
Moreover, if $g(v_i)$ is $v_{i+1}$, we do not visit the vertex between $L(v_i)$ and $v_i$.
After computing $L(v_i)$ for all $v_i \in V_{RC}(T)$, Algorithm~\ref{A1} runs in $O(n)$.
We describe Algorithm~\ref{A2} for Algorithm~\ref{A1} in  greater detail as follows: 


\begin{algorithm}[]
\label{A2}
  \SetAlgoNoEnd
  \caption{Compute a $G_r \subseteq V_{LR}(T) \cup V_{RR}(T)$ of minimum cardinality that guards the $V_{RC}(T)$}
  \KwIn{$T$: terrain}
  \KwOut{$G_r$}
  
  $G_r$ is null;
  
  $V(T')$=$V(T)$;
  
  \For{$v_i \in V_{RC}(T)$ processed from left to right } {
  	\If{$G_r$ is null}{
    $g(v_i)$ be a higher vertex between $L(v_i)$ and $v_{i+1}$;
         
         Add $g(v_i)$ to $G_r$;
         
         Remove the vertics between $v_i$ and $L(v_i)$ from $V(T')$;
    }
   \Else{
    \While {$g(v_i)$ is null}{
    $g_j$ be the rightmost vertex in $G_r \cap V(T')$;
    
     \If{$v_i$ is guarded by $g_j$}
     {
      $g(v_i)$ is $g_j$;
      
      Remove the vertices between $v_i$ and $g_j$ from $V(T')$;
     }
     \ElseIf{$x(g_j) < x(L(v_i))$}
       {
         $g(v_i)$ be a higher vertex between $L(v_i)$ and $v_{i+1}$;
         
         Add $g(v_i)$ to $G_r$;
         
         Remove the vertics between $v_i$ and $L(v_i)$ from $V(T')$;
       }
       \Else
       {Remove $g_j$ from $V(T')$;}
    } 
    }    
    	  
    	}

  \Return $G_r$
\end{algorithm}


\begin{theorem}
Algorithm~\ref{A1} runs in $O(n)$.
\end{theorem}
\begin{proof}
For each $v_i \in V_{RC}(T)$, we examine whether $v_i$ is guarded by $g \in G_r$ from $v_i$ to $g(v_i)$.
If $g(v_i) = g = v_j$, Algorithm~\ref{A1} does not visit $v_k \in \{ v_{k'} \in V(T) \mid j<k'<i \}$.
Assume that $g(v_i)$ is $v_{i+1}$ and $L(v_i)=v_p$, we will not visit $\{v_q \in V(T) \mid p \leq q \leq i \}$ by Lemma~\ref{L(vi) not see}.
Therefore, the algorithm visits the vertex at most $3|V(T)|$ times.
\end{proof}

\begin{figure}
\begin{center}
\includegraphics[scale=0.5]{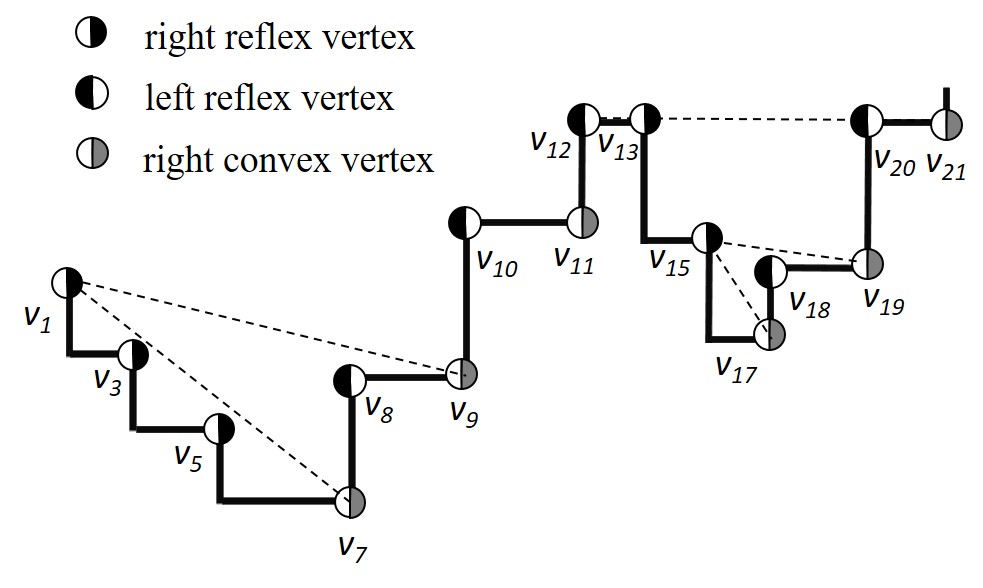}
\caption{The example for Algorithm~\ref{A1}.}
\label{exp}
\end{center}
\end{figure}

An example is provided in Figure~\ref{exp}. Vertex $v_7$ is the leftmost right convex vertex in $T$.
In the first round, algorithm~\ref{A2} visits vertices $v_7, v_6,...,v_1$ and records $g(v_7) = v_1 = L(v_7)$.
In the second round, the algorithm visits vertices $v_9, v_8, v_7$ and $v_1$.
Because $v_1 \in G_r$ and $v_1 = L(v_9)$, the algorithm records $g(v_9) = v_1$.
In the third round, algorithm visits vertices $v_{11}$ and $v_{10}$ because $v_{10} = L(v_{11})$. Because $y(v_{10}) < y(v_{12})$, $g(v_{11})$ is $v_{12}$.
In the fourth round, $g(v_{17})$ is $v_{15}$.
In the fifth round, the algorithm visits vertices $v_{19}, v_{18}, v_{17}$ and $v_{15}$.
Based on $v_{15} \in G_r$ and $\angle v_{19}v_{18}v_{15} < 180\degree$, the algorithm records that $g(v_{19})$ is $v_{15}$.
In the final round, the algorithm vertices $v_{21}, v_{20}, v_{19}, v_{15},...,v_{12}$ and records that $g(v_{21}) = v_{12}$.

\section{Conclusion}
This paper considers the problem of guarding the orthogonal terrain vertex $V(T)$ with the minimum number of vertex guards. 
Our algorithm can determine the minimal cardinality vertex that guards the right (left) convex vertex of $T$.
This paper demonstrates that our algorithm runs in $O(n)$ where $n$ is the number of vertices on $T$.
Therefore, we provide a 2-approximation algorithm for the orthogonal terrain guarding problem in $O(n)$.



\bibliographystyle{elsarticle-num.bst}
\bibliography{bibliography.bib}
\end{document}